\documentclass[journal]{IEEEtran}

\usepackage[utf8]{inputenc}
\usepackage[T1]{fontenc}
\usepackage{amsmath, amsfonts, amssymb, mathtools}
\usepackage{amsthm}
\usepackage{float}
\usepackage{graphicx}
\usepackage{cite}
\usepackage{xcolor}
\usepackage{booktabs} 
\usepackage{tabularx}
\usepackage{hhline}

\usepackage{braket}
\usepackage{tikz}
\usepackage{quantikz} 

\usepackage[ruled,vlined,linesnumbered]{algorithm2e}
\newcommand{\nonl}{\renewcommand{\nl}{\let\nl\oldnl}}

\usepackage[font=small]{caption}
\usepackage{subcaption}

\usepackage{titlesec}
\usepackage{csquotes}
\usepackage{hyperref}
\hypersetup{
    colorlinks=true,
    linkcolor=black,
    citecolor=magenta,
    urlcolor=blue,
}

\newtheorem{theorem}{Theorem}

\definecolor{codegreen}{rgb}{0,0.6,0}
\definecolor{backcolour}{rgb}{0.95,0.95,0.92}

\graphicspath{{./figures/}{./Images/}} 

\begin{document}

\hyphenation{op-tical net-works semi-conduc-tor}

\title{Unsharp Measurement with Adaptive Gaussian POVMs for Quantum-Inspired Image Processing}
\author{Debashis~Saikia, Bikash~K.~Behera, Mayukha~Pal, Prasanta~K.~Panigrahi

\thanks{Debashis Saikia  is with the Department of Physics, Indian Institute of Science Education and Research, Thiruvananthapuram, India; Email: debashis23@iisertvm.ac.in}

\thanks{B.~K. Behera is with Bikash's Quantum (OPC) Pvt. Ltd., Mohanpur, WB, 741246 India and Università degli Studi di Cagliari, Via Is Mirrions, Cagliari, 09123, Italy; Email: bikas.riki@gmail.com}

\thanks{Mayukha~Pal is with the ABB Ability Innovation Center, Asea Brown Boveri Company, Hyderabad 500084, India.Email: mayukha.pal@in.abb.com}

\thanks{Prasanta~K.~Panigrahi is with Center for Quantum Science and Technology, Siksha O Anusandhan University, Bhubaneswar,
India and Department of Physical Sciences, Indian Institute of Science Education and Research (IISER), Kolkata, Mohanpur 741246, West Bengal, India; Email: director.cqst@soa.ac.in}

}

\maketitle
\begin{abstract}
We propose a data-adaptive probabilistic intensity remapping framework for structure-preserving transformation of grayscale images. The suggested method formulates intensity transformation as a continuous, data-driven remapping process, in contrast to traditional histogram-based techniques that rely on hard thresholding and generate piecewise-constant mappings. The image statistics yield representative intensity values, and Gaussian-based weighting methods probabilistically allocate each pixel to several components. Smooth transitions while preserving structural features are achieved by computing the output intensity as an expectation over these components. A smooth transition from soft probabilistic remapping to hard assignment is made possible by the introduction of a nonlinear sharpening parameter $\gamma$ to regulate the degree of localization. This offers clear control over the trade-off between intensity discrimination and smoothing. Furthermore, the resolution of the remapping function is determined by the number of components $k$. When compared to thresholding-based methods, experimental results on standard benchmark images show that the suggested method achieves better structural fidelity and controlled information reduction as measured by PSNR, SSIM, and entropy. Overall, by allowing continuous, probabilistic intensity modifications, the framework provides a robust and efficient substitute for discrete thresholding.
\end{abstract}
\IEEEpeerreviewmaketitle

\section{Introduction}\label{Sec1}

Quantum measurement constitutes the fundamental mechanism through which information about a physical system is extracted. In the conventional formulation of quantum mechanics, measurements are described by projection-valued measures (PVMs) \cite{nielsen2010quantum, preskill1998ph229, vonNeumann1927a}, where each outcome is associated with an orthogonal projector arising from the spectral decomposition of an observable. Such measurements correspond to idealized scenarios in which the system is projected onto an eigenstate of the measured observable, yielding sharp outcomes with well-defined eigenvalues. While this framework provides a complete description for ideal measurements, it becomes restrictive in practical situations where measurements are subject to uncertainty, noise, or partial information extraction. In particular, the requirement of orthogonality and exact eigenvalue resolution limits the ability of PVMs to describe more general measurement processes that arise in realistic quantum systems and information-processing tasks.

To overcome these limitations, the formalism of generalized quantum measurements based on positive operator-valued measures (POVMs) \cite{nielsen2010quantum, barnett2009quantum, peres1990neumark} was developed. In this framework, measurement outcomes are described by a set of positive semidefinite operators $\{E_k\}$ that satisfy the completeness condition $\sum_k E_k = I$, without requiring mutual orthogonality. The probability of obtaining outcome $k$ for a system in state $\rho$ is given by $p_k = \mathrm{Tr}(E_k \rho)$, thereby extending the Born rule to a more general operator setting. Unlike PVMs, POVMs allow measurement operators to overlap, enabling the description of measurements that extract information in a probabilistic and non-projective manner. This increased flexibility makes POVMs particularly suitable for modeling measurement processes in open systems, indirect measurements, and scenarios involving limited resolution or coarse-graining of observable quantities.

Unsharp (or weak) measurements \cite{Busch1998, BUSCH199810, PhysRevD.33.2253, PhysRevA.91.032116, wiseman2009quantum} provide a controlled level of imprecision in the measurement processes, providing a natural generalization of projective measurements in this larger context. Such measurements effectively investigate coarse-grained versions of observables, where each outcome represents contributions from a variety of neighboring eigen-states, as opposed to assigning outcomes to specific eigenvalues. A kernel that distributes weight throughout the spectrum and whose width determines the measurement strength is a useful way to characterize this phenomenon. 
In addition to providing a helpful viewpoint where measurements function as probabilistic transformations of observables rather than just extracting outcomes, such a framework smoothly interpolates between sharp and extremely coarse-grained measurements.

In image processing, where intensity value transformations are crucial, this unsharp measurement model becomes particularly relevant. Most conventional methods operate by changing these values through statistical or kernel-based processes. A grayscale image can be thought of as a distribution over intensity levels. From this angle, it is natural to consider whether a measurement-theoretic framework may be used to analyze such transformations. Quantum mechanical tools can be implemented in a strictly operator-theoretic manner by encoding intensity values in a Hilbert space. In particular, measurement can be interpreted as a mechanism that induces transformations in the data itself rather than just as a way to retrieve information. In this context, the modified intensities emerge as expectation values of the relevant outcomes, and the measurement operators are built from the statistical structure of the image. This offers an alternative perspective on image transformation in which operator-based descriptions and statistical models are integrated into the same framework rather than being handled independently.

\subsection{Related Works and Proposed Approach}

Histogram-based methods such as Multi-Otsu \cite{6313341} and recursive statistical approaches \cite{ARORA2008119} are widely used for grayscale image processing tasks including segmentation and intensity transformation. Multi-Otsu extends the classical Otsu method \cite{4310076, kapur1985new} by selecting multiple thresholds that maximize inter-class variance, thereby partitioning the intensity histogram into discrete regions. Similarly, recursive statistical methods iteratively determine thresholds based on histogram statistics.

Despite their effectiveness, these approaches rely fundamentally on hard partitioning of the intensity space. Consequently, the resulting mappings are piecewise constant, often leading to quantization artifacts and loss of fine structural details. Moreover, such methods lack a probabilistic formulation, making it difficult to model smooth transitions or uncertainty in intensity representation.

To address these limitations, probabilistic and continuous formulations of intensity mapping have been explored, where pixel intensities are transformed using weighted combinations derived from image statistics. However, existing approaches typically lack a principled operator-based framework that unifies normalization, adaptivity, and controllable localization within a single formulation.

In this context, quantum measurement theory provides a natural mathematical framework for probabilistic transformations. In particular, generalized measurements described by positive operator-valued measures (POVMs) allow overlapping, non-orthogonal operators that can model soft assignments across intensity levels. Quantum image processing models such as FRQI \cite{frqi} and NEQR \cite{neqr}, reviewed in \cite{wang2022quantum}, further demonstrate how image information can be embedded into operator-based representations.

Recent work by Barui \textit{et al.} \cite{barui2024novel} utilized Gaussian-based POVMs for threshold selection, where measurement outcomes were used to define segmentation boundaries. While this approach introduces a probabilistic perspective, the role of the POVM remains limited to determining discrete thresholds.

In contrast, the present work formulates probabilistic intensity remapping itself as a quantum measurement process. Specifically, we construct data-adaptive POVM operators from Gaussian models of the intensity histogram, and define the transformed intensity as the expectation value of measurement outcomes. This results in a continuous, structure-preserving mapping in which each pixel contributes probabilistically to multiple intensity components. Furthermore, a nonlinear sharpening parameter enables controlled transition between unsharp (probabilistic) and sharp (projective) measurement regimes, providing explicit control over the trade-off between smoothing and localization. This formulation retains the full operator structure of the POVM and directly utilizes it for intensity transformation rather than threshold selection.

\subsection{Research Gap and Motivation}

However, the scope of existing approaches is still limited despite the advancement of generalized measurement theory. Measurement is typically only utilized as a last stage in the decision-making process, when the POVM assists in determining threshold values rather than serving as a mechanism that modifies the image. As a result, the POVM's whole operator structure is not completely utilized, and the framework still focuses mostly on segmentation rather than broader picture modifications. Moreover, Gaussian models are not incorporated into a formulation where measurement results immediately produce a continuous mapping of pixel intensities, even if they capture the statistical features of the intensity histogram. This makes the shift from discrete decisions to more flexible and seamless modifications difficult.

These limitations highlight a key gap in the existing literature: the absence of a data-adaptive, operator-theoretic framework in which quantum measurement acts as a transformation mechanism derived directly from the statistical structure of image intensities. This motivates the need for a formulation in which measurement is not merely used for decision-making, but serves as a fundamental mechanism for defining continuous, probabilistic transformations of image data.

\subsection{Novelty and Contributions}

The main contributions of this work are summarized as follows:

\begin{itemize}

    \item We formulate probabilistic intensity remapping as a quantum measurement-induced process, rather than a terminal thresholding step. Data-adaptive operators derived from Gaussian intensity models define an unsharp measurement in which each pixel contributes probabilistically to multiple intensity levels.

    \item We reconstruct intensities using the expectation value of measurement outcomes, yielding a continuous remapping function that replaces hard histogram partitioning and preserves structural information through smooth intensity transitions.

    \item The framework is fully data-adaptive, with measurement operators derived directly from the input image (Sec.~\ref{adaptive}), and incorporates a nonlinear sharpening mechanism (Sec.~\ref{Sharpness}) that controls localization. This enables a continuous transition between soft probabilistic remapping and hard assignment behavior, providing an explicit trade-off between smoothing and intensity discrimination.

    \item The proposed formulation is closely related to kernel-based regression methods such as the Nadaraya–Watson estimator \cite{nadaraya1964estimating, watson1964smooth}, where Gaussian functions act as weights and normalization arises naturally from the POVM completeness condition.

    \item From a quantum mechanical perspective, the remapping operation can be interpreted as the expectation value of an observable associated with an unsharp measurement (Sec.~\ref{operator}). This establishes a direct connection between probabilistic intensity modeling and the operator-theoretic framework of quantum measurement.

\end{itemize}

\subsection{Organization}

The remainder of the paper is organized as follows: Sec~\ref{methodology} describes the proposed methodology in detail, from the construction of adaptive Gaussian POVMs to the reconstruction framework. Sec~\ref{results} discusses the experimental results and compares the proposed approach with existing methods. Section~\ref{discussion} provides a theoretical analysis along with a discussion of adaptive behavior and sharpness properties. Finally, Section~\ref{conclusion} concludes the paper.
\section{Methodology}\label{methodology}

\subsection{Problem Formulation}

Let $I : \Omega \subset \mathbb{Z}^2 \rightarrow \mathcal{I}$ denote a grayscale image, where $\mathcal{I} = \{0,1,\dots,255\}$. The objective is to construct a transformation
\begin{equation}
\hat{I}(x,y) = T\big(I(x,y)\big).
\end{equation}
Conventional classical and existing quantum approaches typically realize $T$ via thresholding or histogram partitioning, leading to piecewise-constant mappings with limited ability to capture smooth intensity variations. In contrast, we formulate $T$ as a probabilistic transformation induced by measurement statistics. Specifically, we construct a set of operators $\{E_k\}_{k=1}^K$ and representative intensities $\{\mu_k\}_{k=1}^K$ such that
\begin{equation}
\hat{I}(x,y) = \sum_{k=1}^K \mu_k \, P_k(x,y),
\label{eqn:reconstruction}
\end{equation}
where $P_k(x,y)$ denotes the measurement probability associated with the input intensity. Thus we formulate the problem as a probabilistic intensity remapping framework, where each input intensity is mapped to an output value through measurement-induced probabilities. Unlike threshold-based methods that partition the intensity space, the proposed approach defines a continuous transformation governed by the statistics of generalized measurements.

\begin{figure*}
    \centering
    \includegraphics[width=\linewidth]{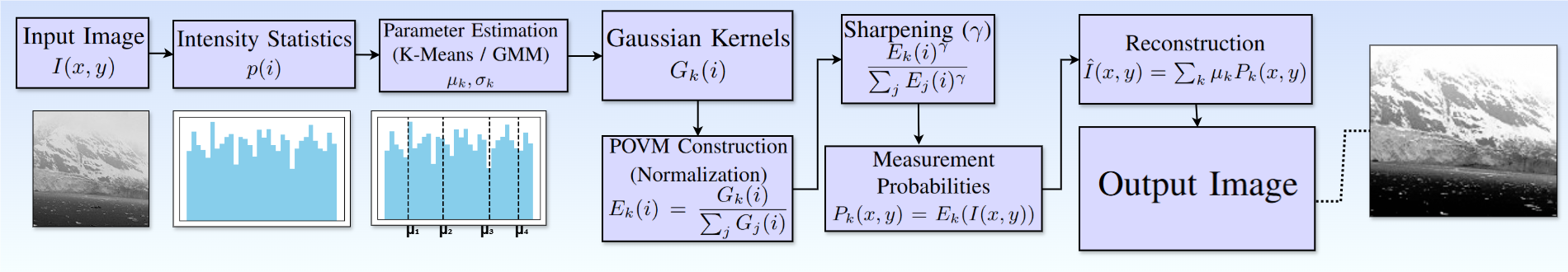}
    \caption{Proposed framework: intensity statistics are used to construct POVMs, followed by sharpening and probabilistic reconstruction.}
    \label{fig:pipeline}
\end{figure*}
\subsection{Image Representation in Hilbert Space}

To enable a measurement-theoretic formulation, grayscale intensities are embedded in a finite-dimensional Hilbert space. Let $\mathcal{H} = \mathbb{C}^{256}$ with orthonormal computational basis
\begin{equation}
\{\,|0\rangle, |1\rangle, \dots, |255\rangle\,\},
\end{equation}
where each basis vector $|i\rangle$ corresponds to an intensity level $i \in \{0,1,\dots,255\}$, establishing a one-to-one mapping between intensities and basis states. For an image $I$ defined over $(x,y)$, each pixel is represented as a pure-state projector
\begin{equation}
\rho_{x,y} = |I(x,y)\rangle \langle I(x,y)|,
\end{equation}
encoding the deterministic intensity in the computational basis. At a global level, the image is described by the diagonal density operator
\begin{equation}
\rho = \sum_{i=0}^{255} p(i)\, |i\rangle \langle i|,
\end{equation}
where $p(i)$ is the normalized intensity histogram. This provides a probabilistic representation of the image and enables the application of quantum measurement operators. Importantly, this embedding is not physical but operator-theoretic, allowing classical data to be processed within a generalized quantum measurement framework.

\subsection{Adaptive Gaussian Construction of POVM}

We construct a family of measurement operators over the intensity Hilbert space that define an unsharp measurement of the intensity observable. The construction is based on Gaussian models derived from the statistical distribution of image intensities, resulting in a data-adaptive set of operators.

\subsubsection{Gaussian Response Functions}

Let $\{\mu_k\}_{k=1}^K$ denote representative intensity values obtained from the image, for example via clustering or statistical estimation. For each $\mu_k$, we define a Gaussian response function over the intensity domain $\mathcal{I}$ as
\begin{equation}
G_k(i) = \exp\left(-\frac{(i - \mu_k)^2}{2\sigma_k^2}\right).
\label{eqn:GRFunction}
\end{equation}
where $\sigma_k$ controls the spread of the $k$-th component.  In the case of uniform spread, a common parameter $\delta$ may be used. These functions define smooth weighting profiles over the intensity domain, assigning higher weights to values close to $\mu_k$ while allowing contributions from neighboring intensities. This naturally implements a coarse-grained measurement consistent with unsharp measurement theory.

\subsubsection{Construction of Measurement Operators}

Using the Gaussian response functions, we define diagonal operators on $\mathcal{H}$:
\begin{equation}
\tilde{E}_k = \sum_{i=0}^{255} G_k(i)\, |i\rangle \langle i|.
\end{equation}
These operators are positive semidefinite but do not necessarily satisfy completeness.

\subsubsection{Normalization and POVM Structure}

To obtain valid measurement operators, we normalize the responses pointwise:
\begin{equation}
E_k(i) = \frac{G_k(i)}{\sum_{j=1}^K G_j(i)}, \quad \forall i \in \mathcal{I},
\end{equation}
and define
\begin{equation}
E_k = \sum_{i=0}^{255} E_k(i)\, |i\rangle \langle i|.
\end{equation}

The resulting operators $\{E_k\}_{k=1}^K$ satisfy positivity and completeness, and therefore constitute a valid POVM.

\subsubsection{Sharpening of Measurement Operators}

To control the degree of measurement sharpness, we introduce a nonlinear transformation parameterized by $\gamma > 0$:
\begin{equation}
E_k(i) \rightarrow \frac{E_k(i)^\gamma}{\sum_{j=1}^K E_j(i)^\gamma}.
\end{equation}
Larger values of $\gamma$ concentrate the distribution around dominant components, approaching projective measurements in the limit $\gamma \to \infty$, while smaller values correspond to smoother measurements.

\subsubsection{Measurement Interpretation}

For a pixel at $(x,y)$ with state $\rho_{x,y}$, the probability of outcome $k$ is
\begin{equation}
P_k(x,y) = \mathrm{Tr}(E_k \rho_{x,y}) = E_k(I(x,y)).
\end{equation}

Thus, the POVM defines an unsharp measurement of intensity, where Gaussian functions act as measurement kernels. Unlike fixed constructions, the operators are derived directly from the image statistics, resulting in a data-adaptive measurement process.

\begin{figure*}
    \centering
    \begin{subfigure}[t]{0.23\linewidth}
        \centering
        \includegraphics[width=\linewidth]{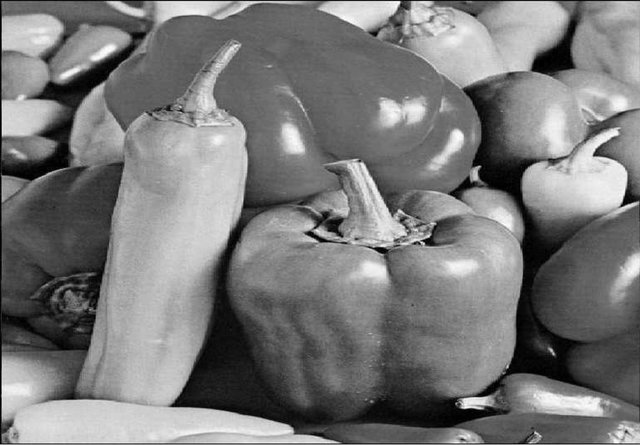}
        \caption{Peppers}
        \label{fig:peppers}
    \end{subfigure}
    \begin{subfigure}[t]{0.16\linewidth}
        \centering
        \includegraphics[width=\linewidth]{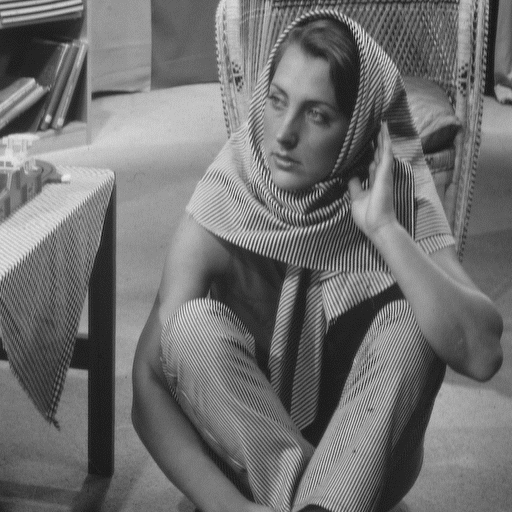}
        \caption{Barbara}
        \label{fig:barbara}
    \end{subfigure}
    \begin{subfigure}[t]{0.16\linewidth}
        \centering
        \includegraphics[width=\linewidth]{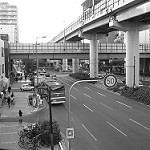}
        \caption{100}
        \label{fig:100}
    \end{subfigure}
    \begin{subfigure}[t]{0.16\linewidth}
        \centering
        \includegraphics[width=\linewidth]{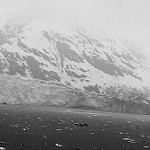}
        \caption{1001}
        \label{fig:1001}
    \end{subfigure}
    \caption{Original Images}
    \label{fig:original_img}
\end{figure*}

\subsection{Image Reconstruction}

Given the constructed POVM, the image transformation is defined through expectation values of measurement outcomes. For a pixel $(x,y)$ with state
\begin{equation}
\rho_{x,y} = \ket{I(x,y)}\bra{I(x,y)},
\end{equation}
the probability of outcome $k$ is given by Eq.~\eqref{eqn:reconstruction}, and the reconstructed value is
\begin{equation}
\hat{I}(x,y) = \sum_{k=1}^K \mu_k \, P_k(x,y),
\end{equation}
which forms a convex combination of representative intensities and thus preserves the valid intensity range.

\subsubsection{Expectation Value Interpretation}\label{operator}

Define the operator
\begin{equation}
A = \sum_{k=1}^K \mu_k E_k.
\end{equation}
Then,
\begin{equation}
\hat{I}(x,y) = \mathrm{Tr}(A \rho_{x,y}),
\end{equation}
showing that the reconstruction is the expectation value of an observable. This establishes a measurement-induced mapping of intensities, replacing discrete decisions with continuous transformations governed by measurement statistics, where the reconstructed intensity represents the average measurement outcome and captures uncertainty in the underlying distribution.

\begin{algorithm}[t]
\caption{Adaptive POVM-Based Measurement and Reconstruction}
\label{alg:povm}

\SetKwInOut{Input}{Input}
\SetKwInOut{Output}{Output}

\Input{Grayscale image $I$, cluster means $\{\mu_k\}_{k=1}^K$, spread parameters $\{\sigma_k\}_{k=1}^K$ (or $\delta$), optional weights $\{w_k\}_{k=1}^K$, sharpening parameter $\gamma$}
\Output{Reconstructed image $\hat{I}$}

\BlankLine
Define intensity domain $\mathcal{I} = \{0,1,\dots,255\}$\;

\BlankLine
\tcp{Construct Gaussian response functions}
\For{$k = 1$ \KwTo $K$}{
    \For{each $i \in \mathcal{I}$}{
        \If{$w_k, \sigma_k$ provided}{
            $G_k(i) \gets w_k \exp\left(-\frac{(i - \mu_k)^2}{2\sigma_k^2}\right)$\;
        }
        \Else{
            $G_k(i) \gets \exp\left(-\frac{(i - \mu_k)^2}{2\delta^2}\right)$\;
        }
    }
}

\BlankLine
Form response matrix $E \in \mathbb{R}^{K \times 256}$\;

\BlankLine
\tcp{Normalize (POVM condition)}
\For{each $i \in \mathcal{I}$}{
    \For{$k = 1$ \KwTo $K$}{
        $E_k(i) \gets \frac{G_k(i)}{\sum_{j=1}^K G_j(i)}$\;
    }
}

\BlankLine
\tcp{Apply sharpening}
\For{each $i \in \mathcal{I}$}{
    \For{$k = 1$ \KwTo $K$}{
        $E_k(i) \gets \frac{E_k(i)^\gamma}{\sum_{j=1}^K E_j(i)^\gamma}$\;
    }
}

\BlankLine
\tcp{Reconstruction}
\For{each pixel $(x,y)$}{
    \For{$k = 1$ \KwTo $K$}{
        $P_k(x,y) \gets E_k(I(x,y))$\;
    }
    $\hat{I}(x,y) \gets \sum_{k=1}^K \mu_k \, P_k(x,y)$\;
}

\Return $\hat{I}$\;
\end{algorithm}

Fig~\ref{fig:pipeline} represents the proposed probabilistic framework. The input image is represented through its intensity statistics, which are used to construct Gaussian kernels and corresponding POVM elements. A sharpening transformation controls measurement localization, and the final image is obtained via probabilistic reconstruction as an expectation value. And the overall procedure of the proposed framework is summarized in Algorithm~\ref{alg:povm}.




\section{Results}\label{results}

\subsection{Experimental Setup}

\subsubsection{Datasets}

The proposed framework is evaluated on a set of images, namely \textit{Peppers} \cite{lena_peppers_barbara}, \textit{Barbara} \cite{lena_peppers_barbara}, \textit{100} \cite{landscape_colorization_kaggle}, and \textit{1001} \cite{landscape_colorization_kaggle}, as shown in Fig.~\ref{fig:original_img}. To facilitate direct embedding into the Hilbert space $\mathcal{H} \subseteq \mathbb{C}^{256}$(Sec.~\ref{methodology}), all images are transformed to grayscale in the interval $[0,\ 255]$. The selected images offer a variety of intensity histogram profiles because they include urban settings, portraits and natural scenes.  This diversity ensures a comprehensive evaluation of the robustness of the proposed method.

\subsubsection{Estimation of Representative Intensity Values}

The construction of the measurement operators requires a set of representative intensities $\{\mu_k\}_{k=1}^K$ which capture the statistical structure of the image. These are obtained through data-driven estimation rather than predefined selection. In this work, we used two approaches: (i) K-Means clustering on pixel intensities, where cluster centers define $\{\mu_k\}$, and (ii) Gaussian Mixture Model (GMM) fitting to the intensity histogram, where component means define $\{\mu_k\}$ and variances provide the spread parameters $\{\sigma_k\}$. In the GMM-based approach, $\{\sigma_k\}$ are derived from component covariances, enabling adaptive behavior, while in the KMeans-based method, a uniform spread $\delta$ is used. These procedures enable data-driven estimation of the underlying intensity distribution; the subsequent operator construction and transformation remain entirely within the operator-theoretic framework as described in Sec.~\ref{methodology}. This estimation of $\{\mu_k\}$ serves as a data-adaptive mechanism for defining measurement operators rather than a learning-based transformation.

\subsubsection{Parameters and Hyperparameters}

The proposed framework is governed by key parameters which includes the number of components $K$, the variance (spread) $\{\sigma_k\}$ (or $\delta$), and the sharpening parameter $\gamma$. The parameter $K$ controls the resolution of the intensity representation, with larger values giving a better partition of the intensity space. On the other hand, the spread parameter determines the width of the Gaussian response functions (Eq.~\ref{eqn:GRFunction}). The sharpening parameter $\gamma$ controls measurement localization, where smaller $\gamma$ values corresponds to the unsharp region and larger values approach a projective regime. GMM parameters are estimated via expectation-maximization, while KMeans determines cluster centers through variance minimization. Together, these parameters enable controlled exploration of the trade-off between smoothing and localization.

\begin{figure*}
\centering


\begin{subfigure}[t]{0.19\textwidth}
    \centering
    \includegraphics[width=\linewidth]{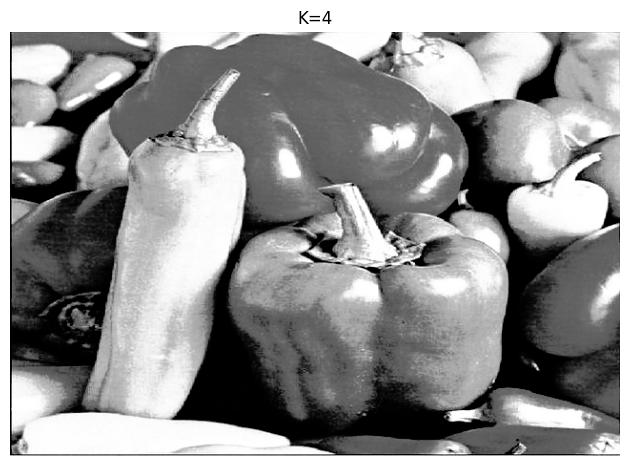}
    \caption{Proposed K-Means}
    \label{fig:peppers_kmeans}
\end{subfigure}\hfill
\begin{subfigure}[t]{0.19\textwidth}
    \centering
    \includegraphics[width=\linewidth]{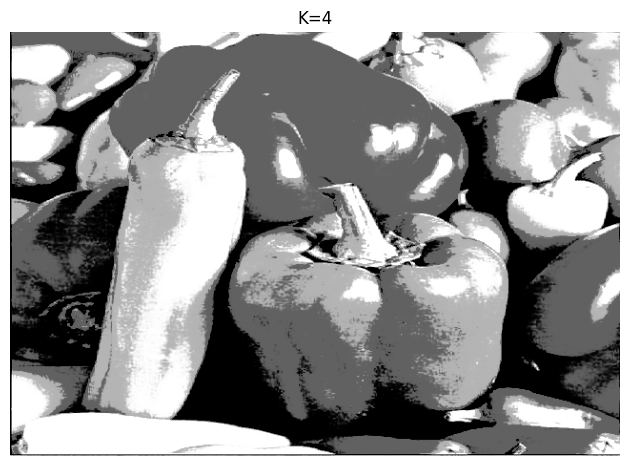}
    \caption{Proposed GMM}
    \label{fig:peppers_gmm}
\end{subfigure}\hfill
\begin{subfigure}[t]{0.19\textwidth}
    \centering
    \includegraphics[width=\linewidth]{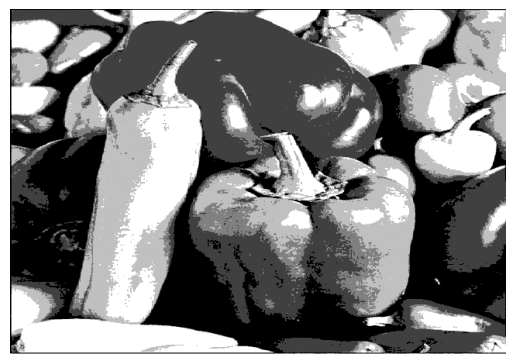}
    \caption{Unsharp Measurement}
    \label{fig:peppers_unsharp}
\end{subfigure}\hfill
\begin{subfigure}[t]{0.19\textwidth}
    \centering
    \includegraphics[width=\linewidth]{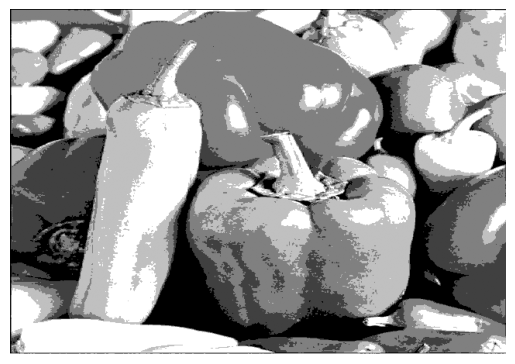}
    \caption{Multi-Otsu}
    \label{fig:peppers_otsu}
\end{subfigure}\hfill
\begin{subfigure}[t]{0.19\textwidth}
    \centering
    \includegraphics[width=\linewidth]{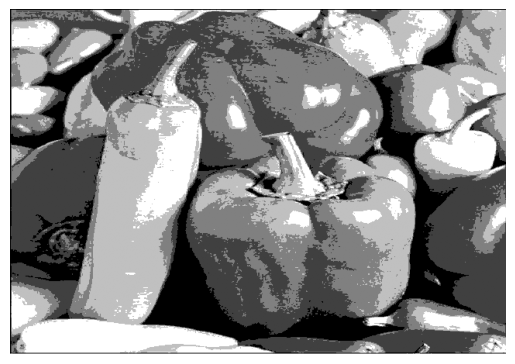}
    \caption{Fast Statistical Recursive}
    \label{fig:peppers_fast}
\end{subfigure}

\begin{subfigure}[t]{0.19\textwidth}
    \centering
    \includegraphics[width=\linewidth]{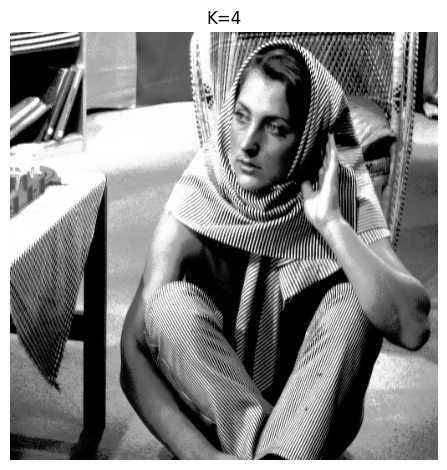}
    \caption{Proposed K-Means}
    \label{fig:barbara_kmeans}
\end{subfigure}\hfill
\begin{subfigure}[t]{0.19\textwidth}
    \centering
    \includegraphics[width=\linewidth]{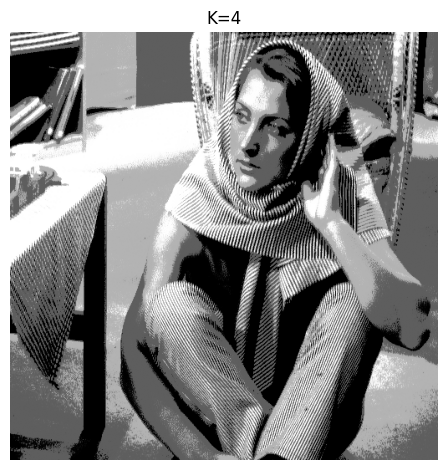}
    \caption{Proposed GMM}
    \label{fig:barbara_gmm}
\end{subfigure}\hfill
\begin{subfigure}[t]{0.19\textwidth}
    \centering
    \includegraphics[width=\linewidth]{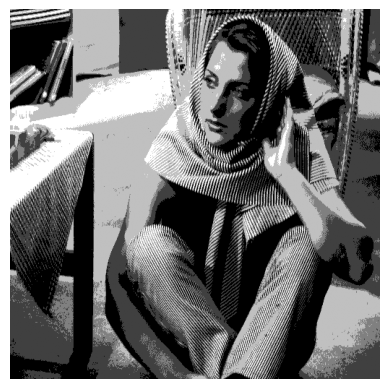}
    \caption{Unsharp Measurement}
    \label{fig:barbara_unsharp}
\end{subfigure}\hfill
\begin{subfigure}[t]{0.19\textwidth}
    \centering
    \includegraphics[width=\linewidth]{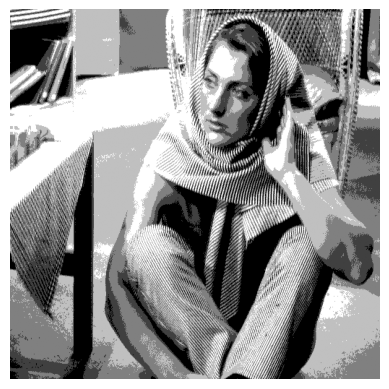}
    \caption{Multi-Otsu}
    \label{fig:barbara_otsu}
\end{subfigure}\hfill
\begin{subfigure}[t]{0.19\textwidth}
    \centering
    \includegraphics[width=\linewidth]{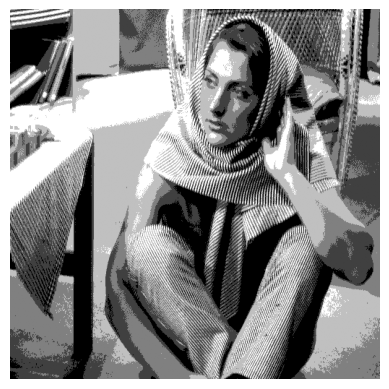}
    \caption{Fast Statistical Recursive}
    \label{fig:barbara_fast}
\end{subfigure}

\begin{subfigure}[t]{0.24\textwidth}
    \centering
    \includegraphics[width=\linewidth]{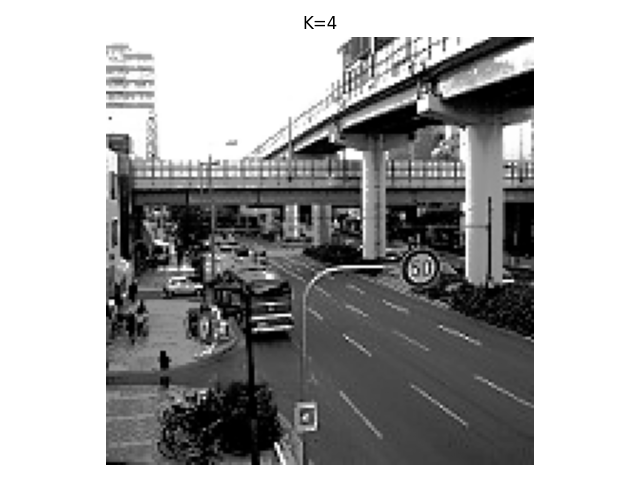}
    \caption{Proposed K-Means}
    \label{fig:img100_kmeans}
\end{subfigure}\hfill
\begin{subfigure}[t]{0.24\textwidth}
    \centering
    \includegraphics[width=\linewidth]{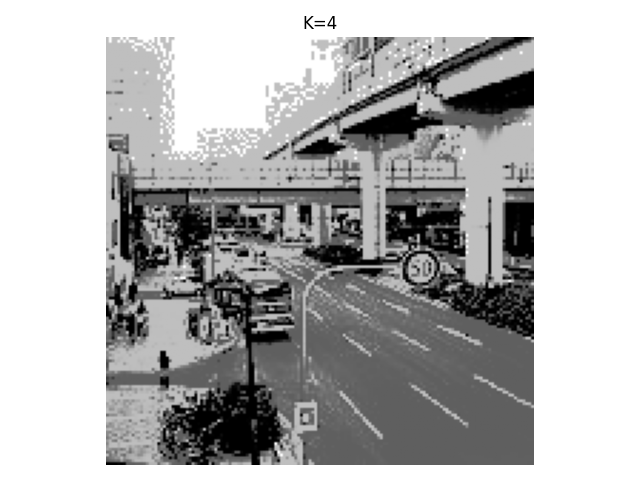}
    \caption{Proposed GMM}
    \label{fig:img100_gmm}
\end{subfigure}\hfill
\begin{subfigure}[t]{0.17\textwidth}
    \centering
    \includegraphics[width=\linewidth]{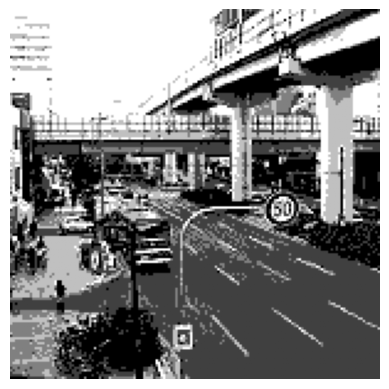}
    \caption{Unsharp Measurement}
    \label{fig:img100_unsharp}
\end{subfigure}\hfill
\begin{subfigure}[t]{0.17\textwidth}
    \centering
    \includegraphics[width=\linewidth]{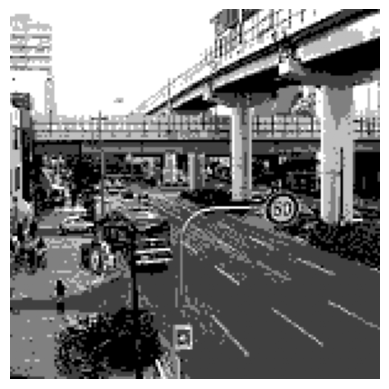}
    \caption{Multi-Otsu}
    \label{fig:img100_otsu}
\end{subfigure}\hfill
\begin{subfigure}[t]{0.17\textwidth}
    \centering
    \includegraphics[width=\linewidth]{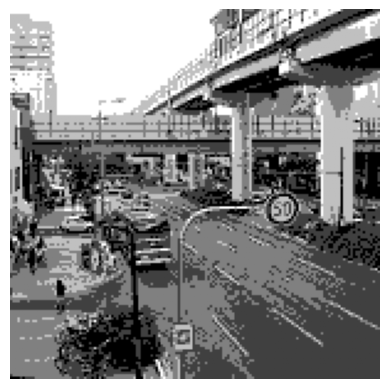}
    \caption{Fast Statistical Recursive}
    \label{fig:img100_fast}
\end{subfigure}

\begin{subfigure}[t]{0.24\textwidth}
    \centering
    \includegraphics[width=\linewidth]{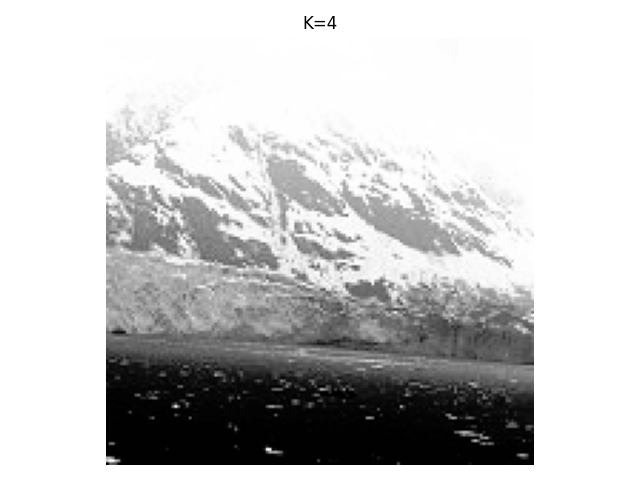}
    \caption{Proposed K-Means}
    \label{fig:img1001_kmeans}
\end{subfigure}\hfill
\begin{subfigure}[t]{0.24\textwidth}
    \centering
    \includegraphics[width=\linewidth]{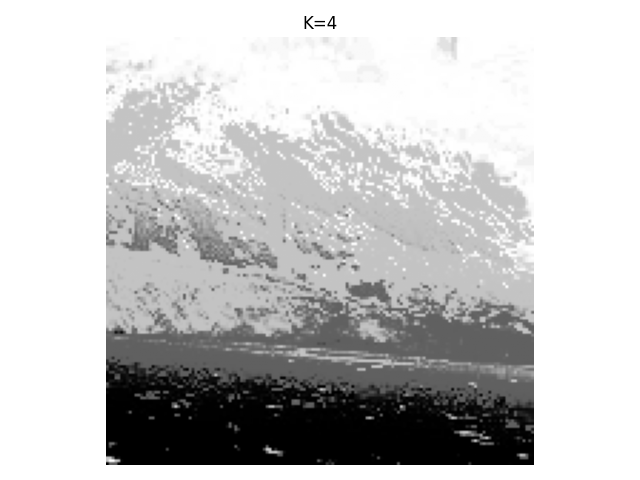}
    \caption{Proposed GMM}
    \label{fig:img1001_gmm}
\end{subfigure}\hfill
\begin{subfigure}[t]{0.17\textwidth}
    \centering
    \includegraphics[width=\linewidth]{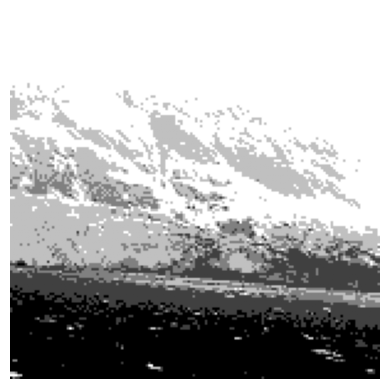}
    \caption{Unsharp Measurement}
    \label{fig:img1001_unsharp}
\end{subfigure}\hfill
\begin{subfigure}[t]{0.17\textwidth}
    \centering
    \includegraphics[width=\linewidth]{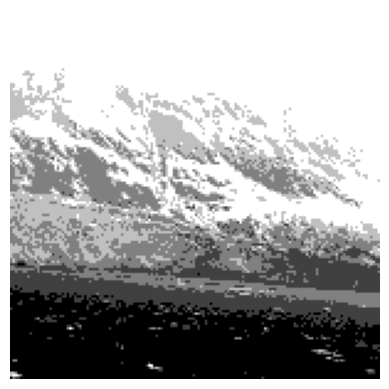}
    \caption{Multi-Otsu}
    \label{fig:img1001_otsu}
\end{subfigure}\hfill
\begin{subfigure}[t]{0.17\textwidth}
    \centering
    \includegraphics[width=\linewidth]{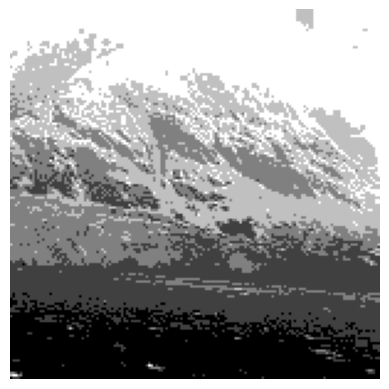}
    \caption{Fast Statistical Recursive}
    \label{fig:img1001_fast}
\end{subfigure}

\caption{Comparison of reconstructed images obtained using Proposed K-Means, Proposed GMM, Unsharp Measurement, Multi-Otsu, and Fast Statistical Recursive methods (from left to right). Rows correspond to different datasets: (top to bottom) Lena, Peppers, Barbara, Image 100, and Image 1001.}
\label{fig:all_comparisons}

\end{figure*}

\subsection{Visual Results Comparison}

In Fig.~\ref{fig:all_comparisons}, reconstructed pictures using GMM- and KMeans-based POVMs are compared with unsharp measurement, Multi-Otsu, and fast statistical recursive methods. For consistency, $k=4$ Gaussian centers are employed for all approaches. The suggested techniques maintain shape and shading for the Peppers image (Figs.~\ref{fig:peppers_kmeans}--\ref{fig:peppers_fast}), while unsharp measurement decreases contrast, Multi-Otsu creates piecewise-constant artifacts, and the fast statistical method distorts structure. The suggested approaches preserve high-frequency textures for the Barbara image (Figs.~\ref{fig:barbara_kmeans}--\ref{fig:barbara_fast}); KMeans maintains structured patterns like stripes, while GMM produces smoother outputs; unsharp is unable to capture texture, Multi-Otsu introduces quantization artifacts, and the fast statistical method loses structural detail. For the Image 100 urban scene (Figs.~\ref{fig:img100_kmeans}--\ref{fig:img100_fast}), the proposed methods preserve key structures, with KMeans enhancing edges and GMM providing smoother transitions, while unsharp and Multi-Otsu introduce segmentation artifacts and the fast statistical method reduces contrast and fine detail. For Image 1001 (Figs.~\ref{fig:img1001_kmeans}--\ref{fig:img1001_fast}), the proposed methods preserve gradients and homogeneous regions with clear intensity separation; KMeans yields sharper outputs and GMM produces more consistent smoothing, whereas unsharp leads to oversmoothing, Multi-Otsu causes excessive discretization, and the fast statistical method degrades both smooth and structured regions. Overall, the proposed methods maintain structural fidelity while avoiding artifacts introduced by thresholding-based approaches.

\begin{figure*}
    \centering
    
    \begin{subfigure}[t]{0.3\linewidth}
        \centering
        \includegraphics[width=\linewidth]{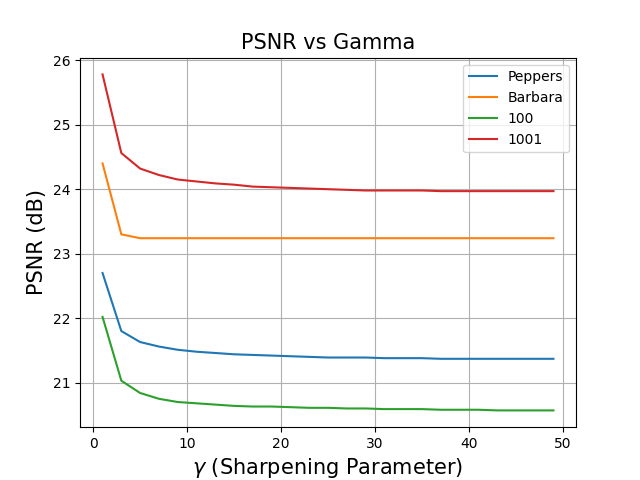}
        \caption{PSNR vs $\gamma$ (K-Means)}
    \end{subfigure}
    \hfill
    \begin{subfigure}[t]{0.3\linewidth}
        \centering
        \includegraphics[width=\linewidth]{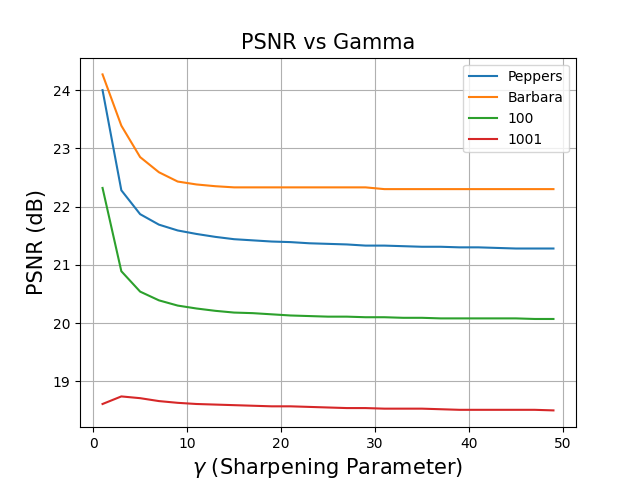}
        \caption{PSNR vs $\gamma$ (GMM)}
    \end{subfigure}
    \hfill
    \begin{subfigure}[t]{0.3\linewidth}
        \centering
        \includegraphics[width=\linewidth]{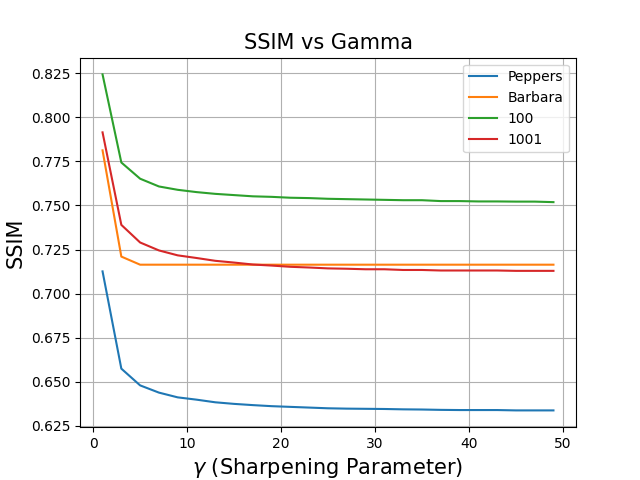}
        \caption{SSIM vs $\gamma$ (K-Means)}
    \end{subfigure}

    \begin{subfigure}[t]{0.3\linewidth}
        \centering
        \includegraphics[width=\linewidth]{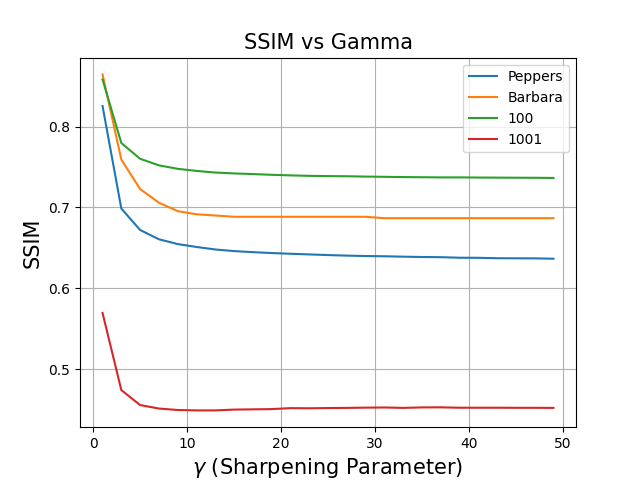}
        \caption{SSIM vs $\gamma$ (GMM)}
    \end{subfigure}
    \hfill
    \begin{subfigure}[t]{0.3\linewidth}
        \centering
        \includegraphics[width=\linewidth]{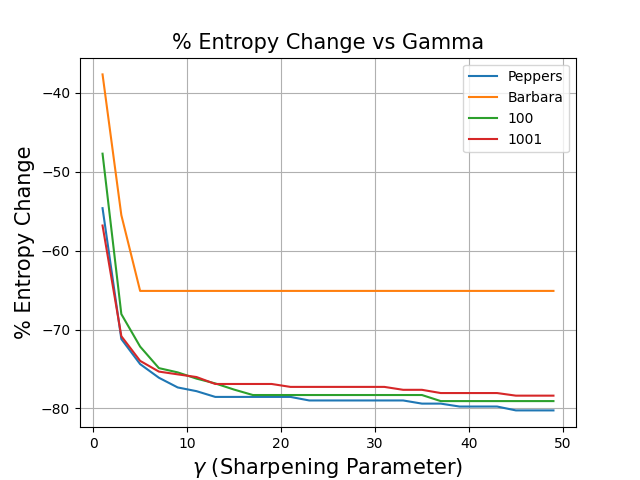}
        \caption{$\Delta$Entropy vs $\gamma$ (K-Means)}
    \end{subfigure}
    \hfill
    \begin{subfigure}[t]{0.3\linewidth}
        \centering
        \includegraphics[width=\linewidth]{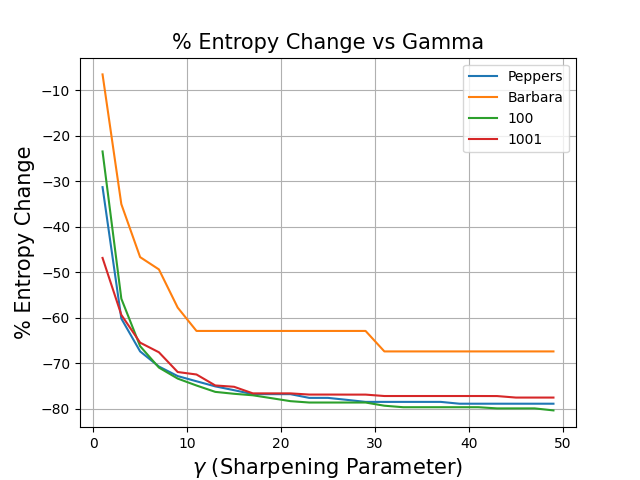}
        \caption{$\Delta$Entropy vs $\gamma$ (GMM)}
    \end{subfigure}

    \begin{subfigure}[t]{0.3\linewidth}
        \centering
        \includegraphics[width=\linewidth]{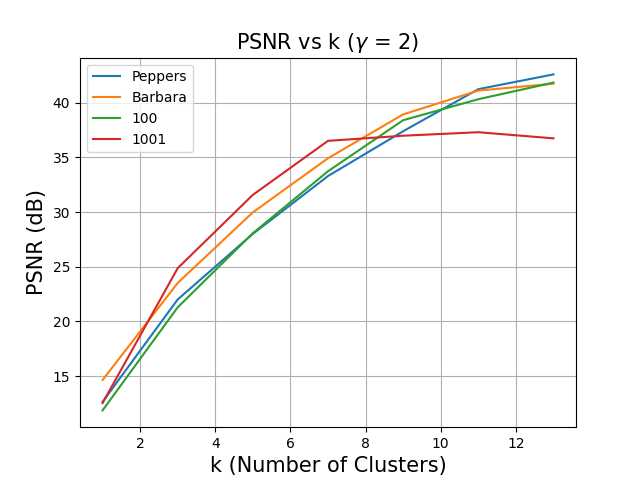}
        \caption{PSNR vs $k$ (K-Means)}
    \end{subfigure}
    \hfill
    \begin{subfigure}[t]{0.3\linewidth}
        \centering
        \includegraphics[width=\linewidth]{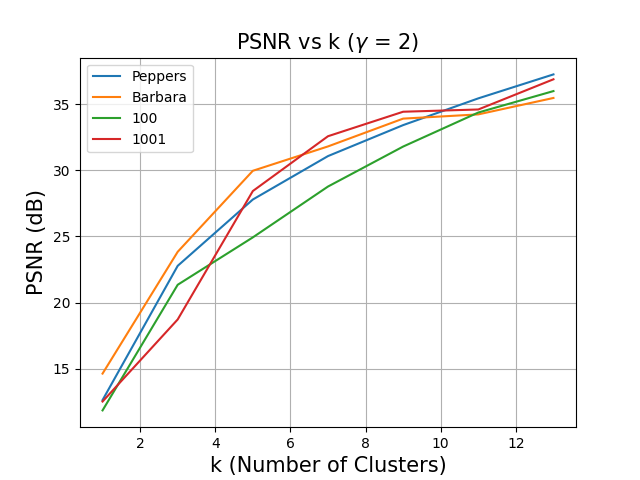}
        \caption{PSNR vs $k$ (GMM)}
    \end{subfigure}
    \hfill
    \begin{subfigure}[t]{0.3\linewidth}
        \centering
        \includegraphics[width=\linewidth]{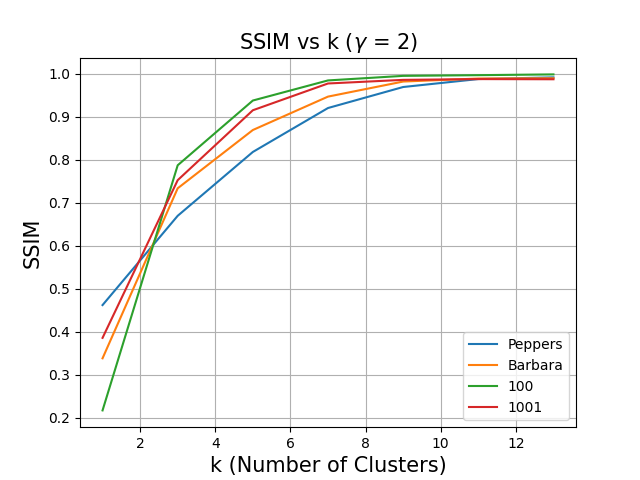}
        \caption{SSIM vs $k$ (K-Means)}
    \end{subfigure}

    \begin{subfigure}[t]{0.3\linewidth}
        \centering
        \includegraphics[width=\linewidth]{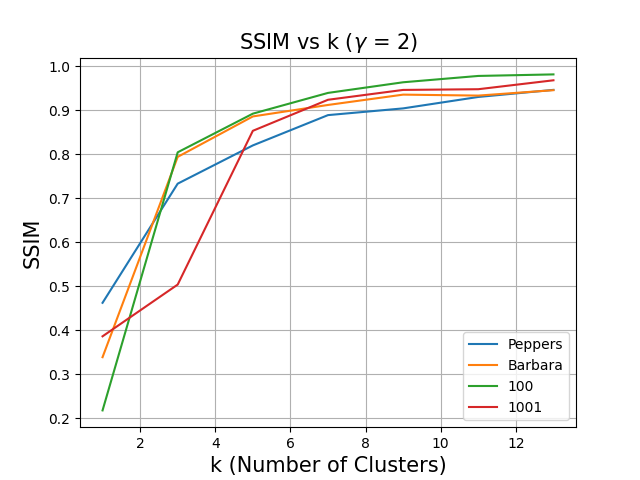}
        \caption{SSIM vs $k$ (GMM)}
    \end{subfigure}
    \hfill
    \begin{subfigure}[t]{0.3\linewidth}
        \centering
        \includegraphics[width=\linewidth]{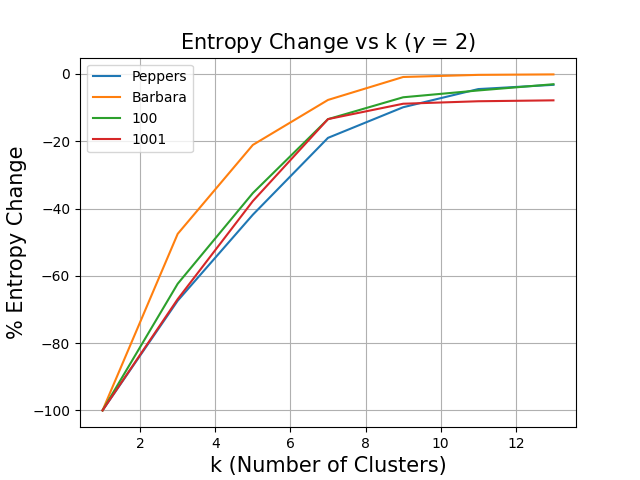}
        \caption{$\Delta$Entropy vs $k$ (K-Means)}
    \end{subfigure}
    \hfill
    \begin{subfigure}[t]{0.3\linewidth}
        \centering
        \includegraphics[width=\linewidth]{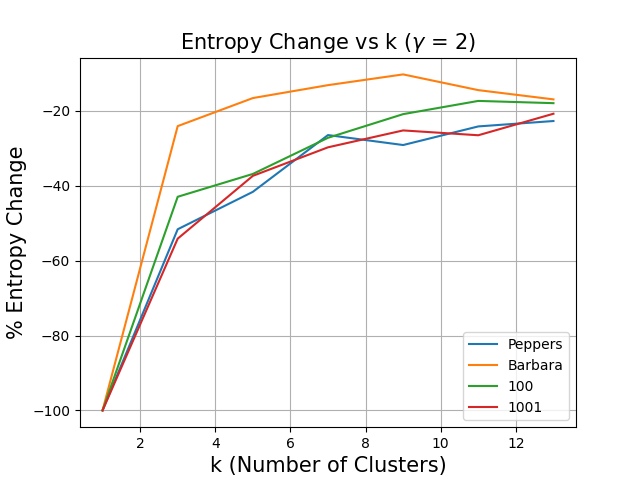}
        \caption{$\Delta$Entropy vs $k$ (GMM)}
    \end{subfigure}
    
    \caption{Comprehensive analysis of PSNR, SSIM, and $\% \Delta \mathrm{Entropy}$ with respect to $\gamma$ and $k$ for the proposed K-Means and GMM-based algorithms.}
    \label{fig:combined_analysis}
\end{figure*}

\subsection{Quantitative Results Comparison}

\begin{table*}[t]
\centering
\caption{Performance comparison across multiple datasets.}
\label{tab:all_results}

\begin{subtable}[t]{0.45\textwidth}
\centering
\caption{Performance Comparison for Peppers Image}
\label{tab:peppers}
\begin{tabular}{lcccc}
\toprule
Algorithm & PSNR & SSIM & $\Delta$Entropy (\%) & Time (s) \\
\midrule
Fast Statistical & 19.6173 & 0.5934 & -71.1138 & 0.0156 \\
Multi-Otsu       & 17.2776 & 0.6291 & -71.1035 & 5.2597 \\
Unsharp Measure  & 20.2134 & 0.5958 & -68.2602 & 0.3141 \\
\textbf{Proposed (GMM)}   & 27.7900 & 0.8203 & -41.6400 & 3.3264 \\
\textbf{Proposed (KMeans)}& \textbf{31.7500} & \textbf{0.9567} & -9.4800  & 0.8344 \\
\bottomrule
\end{tabular}
\end{subtable}
\hfill
\begin{subtable}[t]{0.45\textwidth}
\centering
\caption{Performance Comparison for Barbara Image}
\label{tab:barbara}
\begin{tabular}{lcccc}
\toprule
Algorithm & PSNR & SSIM & $\Delta$Entropy (\%) & Time (s) \\
\midrule
Fast Statistical & 16.2648 & 0.6044 & -39.2213 & 0.0466 \\
Multi-Otsu       & 14.9766 & 0.5933 & -49.1687 & 3.0221 \\
Unsharp Measure  & 19.7673 & 0.6438 & -47.9933 & 0.1149 \\
\textbf{Proposed (GMM)}   & 29.9600 & 0.8862 & -16.6000 & 3.3530 \\
\textbf{Proposed (KMeans)}& \textbf{34.6400} & \textbf{0.9658} & -1.8200  & 0.6704 \\
\bottomrule
\end{tabular}
\end{subtable}

\vspace{0.6em}

\begin{subtable}[t]{0.45\textwidth}
\centering
\caption{Performance Comparison for Image 100}
\label{tab:100}
\begin{tabular}{lcccc}
\toprule
Algorithm & PSNR & SSIM & $\Delta$Entropy (\%) & Time (s) \\
\midrule
Fast Statistical & 20.5148 & 0.7930 & -70.2130 & 0.0087 \\
Multi-Otsu       & 17.8657 & 0.7792 & -70.2385 & 5.0683 \\
Unsharp Measure  & 20.6549 & 0.8255 & -66.6448 & 0.1564 \\
\textbf{Proposed (GMM)}  & 24.9300 & 0.8925 & -36.8600 & 1.8837 \\
\textbf{Proposed (KMeans)}& \textbf{31.3500} & \textbf{0.9805} & -9.6200  & 0.1087 \\
\bottomrule
\end{tabular}
\end{subtable}
\hfill
\begin{subtable}[t]{0.45\textwidth}
\centering
\caption{Performance Comparison for Image 1001}
\label{tab:1001}
\begin{tabular}{lcccc}
\toprule
Algorithm & PSNR & SSIM & $\Delta$Entropy (\%) & Time (s) \\
\midrule
Fast Statistical & 16.4303 & 0.5619 & -67.4542 & 0.0108 \\
Multi-Otsu       & 16.0354 & 0.6292 & -69.7922 & 3.3858 \\
Unsharp Measure  & 18.6952 & 0.7145 & -66.6559 & 0.1785 \\
\textbf{Proposed (GMM)}   & 28.4300 & 0.8535 & -37.3700 & 2.5751 \\
\textbf{Proposed (KMeans)}& \textbf{31.9000} & \textbf{0.9754} & -11.4100 & 0.1170 \\
\bottomrule
\end{tabular}
\end{subtable}

\end{table*}
To quantitatively evaluate the proposed probabilistic intensity remapping framework, we employ Peak Signal-to-Noise Ratio (PSNR) \cite{korhonen2012peak}, Structural Similarity Index Measure (SSIM) \cite{ssim}, and the percentage change in Shannon entropy. Since the objective is not exact reconstruction but structure-preserving intensity remapping, these metrics are interpreted as measures of distortion, structural consistency, and information retention relative to the input image.

PSNR is defined as
\begin{equation}
\mathrm{PSNR} = 10 \log_{10} \left( \frac{MAX^2}{\mathrm{MSE}} \right),
\end{equation}
where $MAX$ is the maximum pixel value and MSE is the mean squared difference between the input and remapped images. In this context, PSNR quantifies the level of distortion introduced by the remapping.

SSIM is defined by combining luminance, contrast, and structural information:
\begin{equation}
\mathrm{SSIM}(x,y) = \frac{(2\mu_x \mu_y + C_1)(2\sigma_{xy} + C_2)}{(\mu_x^2 + \mu_y^2 + C_1)(\sigma_x^2 + \sigma_y^2 + C_2)},
\end{equation}
where $\mu_x, \mu_y$ are mean intensities, $\sigma_x^2, \sigma_y^2$ and $\sigma_{xy}$ are variances and covariance respectively. SSIM evaluates the preservation of structural information, particularly edges and textures, under the remapping.

Shannon entropy measures the information content of the intensity distribution:
\begin{equation}
H = - \sum_{i=1}^{N} p_i \log_2 p_i,
\end{equation}
where $p_i$ is the probability of occurrence of intensity level $i$. The percentage change in entropy is used to quantify information retention, with smaller reductions indicating better preservation of intrinsic image content. Together, these metrics provide a complementary evaluation in terms of distortion control, structural preservation, and information retention, which are central to probabilistic intensity remapping.

The proposed approaches consistently outperform baseline methods in terms of reconstruction fidelity, as demonstrated by the PSNR values in Tables~\ref{tab:peppers}-\ref{tab:1001}. While the GMM-based strategy also increases performance with PSNR often in the range of 24–31, the K-Means-based method consistently achieves the highest PSNR, achieving a peak value of 35.21 for Lena and remaining over 31 in most cases. Conversely, the fast statistical recursive method~\cite{ARORA2008119}, the unsharp measurement-based approach~\cite{barui2024novel}, and Multi-Otsu~\cite{6313341} produce much lower PSNR values, typically below 21. Similar results are seen for SSIM, where the proposed approaches provide significant structural preservation, with values above 0.95 for KMeans and above 0.85 for GMM (e.g., 0.9711 for Lena and 0.9805 for Image 100), as opposed to baseline methods, which range between 0.56 and 0.82.

The percentage change in Shannon entropy further highlights the advantage of the proposed framework. Conventional methods such as Multi-Otsu~\cite{6313341} and the statistical recursive approach~\cite{ARORA2008119} result in substantial entropy reductions (often exceeding 60--70\%), whereas the proposed methods exhibit significantly lower entropy loss. In particular, the KMeans-based method maintains entropy reduction within approximately 2--12\%, indicating better preservation of intrinsic image information, while the unsharp measurement-based method~\cite{barui2024novel} shows noticeably higher loss.

From a computational perspective, the KMeans-based method remains efficient, with execution times generally below one second, while the GMM-based approach incurs higher cost due to expectation-maximization but remains competitive with Multi-Otsu~\cite{6313341}. Although the fast statistical recursive method~\cite{ARORA2008119} is computationally efficient, it does so at the expense of reconstruction quality, and the unsharp measurement-based approach~\cite{barui2024novel} fails to achieve comparable performance. 

\begin{table*}[h]
\centering
\caption{Quantitative comparison on 100 randomly selected BSD500 images (mean $\pm$ standard deviation). Proposed methods evaluated with $k=8$ and $\gamma=2$ as probabilistic intensity remapping models.}
\label{tab:bsd500}
\begin{tabular}{lccc}
\hline
\toprule
Algorithm & PSNR & SSIM & $\Delta$Entropy (\%)\\
\hline
Gaussian Filter \cite{deng1993adaptive}
& 28.6607 $\pm$ 3.6590 & 0.8594 $\pm$ 0.0677 & -1.1165 $\pm$ 1.3012 \\
Bilateral Filter \cite{tomasi1998bilateral} & 30.3301 $\pm$ 3.1104 &  0.8700 $\pm$ 0.0567 & -1.1259 $\pm$ 1.1271 \\
BM3D \cite{danielyan2011bm3d} & 29.3568 $\pm$ 3.2318 & 0.8168 $\pm$ 0.0941 &-1.4039 $\pm$ 1.3729 \\
CLAHE \cite{reza2004realization} & 20.2173 $\pm$ 2.1111 & 0.8373 $\pm$ 0.0438 & 6.4902 $\pm$ 4.5446 \\
FastNLM \cite{buades2011non} & 34.7468 $\pm$ 2.2081 & 0.8911 $\pm$ 0.0634 &-1.3443 $\pm$ 1.2013 \\
\textbf{Proposed (GMM)} & 33.4993 $\pm$ 1.9285 & 0.9483 $\pm$ 0.0208 & \textbf{-25.9031 $\pm$ 4.5221} \\
\textbf{Proposed (K-Means)} & \textbf{34.7968 $\pm$ 2.8695} & \textbf{0.9783 $\pm$ 0.0164} & -6.8283 $\pm$ 3.1276 \\
\hline
\end{tabular}
\end{table*}
The results in Table~\ref{tab:bsd500} demonstrate that the proposed POVM-based methods consistently achieve superior performance on the BSD500 dataset \cite{MartinFTM01}. In particular, the KMeans-based model attains the highest PSNR and SSIM with low variance, indicating robust reconstruction fidelity and strong structural preservation across diverse images. Conventional methods such as Gaussian and bilateral filtering exhibit only marginal entropy changes, but at the cost of reduced structural quality, while methods such as BM3D and FastNLM show limited improvement in perceptual similarity. In contrast, the proposed framework enables a controlled trade-off between structural preservation and redundancy reduction. The KMeans-based model maintains high SSIM (0.9783) with moderate entropy reduction, preserving fine image details. Notably, the GMM-based model exhibits a significantly larger entropy reduction ($-25.90\%$), indicating effective suppression of redundant intensity variations rather than arbitrary information loss. This behavior is advantageous in applications requiring compact representation, such as image compression, while still maintaining high structural consistency (SSIM $=0.9483$). Overall, these results highlight the flexibility of the proposed probabilistic reconstruction framework in balancing fidelity and compactness through data-adaptive measurement design.
\section{Discussion}\label{discussion}

\subsection{Adaptive behavior}\label{adaptive}
We assume that the representative values $\{\mu_k\}$ provide a sufficiently dense coverage of the intensity space, in the sense that for each intensity level $i$, there exists a $\mu_k$ such that $|i - \mu_k| \leq \epsilon_K$, where $\epsilon_K \to 0$ as $K \to \infty$.

\begin{theorem}[Consistency of Adaptive POVM Reconstruction]
Let $I(x,y) \in \{0,1,\dots,255\}$ be a grayscale image and let $\{\mu_k\}_{k=1}^K$ be a set of representative intensities obtained from a statistical model (e.g., a Gaussian mixture model) such that for each intensity level $i$, there exists $\mu_k$ satisfying
\begin{equation}
    |\mu_k - i| \leq \epsilon_K,
\end{equation}
where $\epsilon_K \to 0$ as $K \to \infty$. Then, for fixed $\gamma$, the reconstruction satisfies
\begin{equation}    
    |\hat{I}(x,y) - I(x,y)| \to 0
    \quad \text{as } K \to \infty,
\end{equation}
for all $(x,y)$.
\end{theorem}

This condition ensures that the discrete set $\{\mu_k\}$ provides an increasingly refined approximation of the intensity domain. The adaptive nature of the proposed measurement framework is governed by the number of Gaussian components $K$. As indicated by the theorem, increasing $K$ improves reconstruction fidelity by refining the representation of the intensity space. For small $K$, the induced POVM yields a coarse approximation, leading to stronger averaging over neighboring intensities and smoother outputs with reduced structural detail. As $K$ increases, the Gaussian components provide a finer coverage of the intensity domain, enabling measurement probabilities to better capture local variations and thereby enhance contrast and structural fidelity, consistent with Fig.~\ref{fig:combined_analysis}. This behavior follows from the reconstruction being a convex combination of representative intensities weighted by measurement probabilities, where larger $K$ increases expressive power. The effect of $K$ should also be considered alongside the sharpening parameter $\gamma$, which controls measurement localization.

\subsection{Sharpness Theorem}\label{Sharpness}

\begin{theorem}[Sharpness Theorem]
Let $\{E_k\}_{k=1}^K$ be a POVM constructed from Gaussian response functions given by Eq.~\ref{eqn:GRFunction}, and let the sharpened coefficients be defined as
\begin{equation}
    \tilde{E}_k(i) = \frac{E_k(i)^\gamma}{\sum_{j=1}^K E_j(i)^\gamma},  \quad \gamma > 0.
\end{equation}
Then, for each fixed $i$, as $\gamma \to \infty$, the coefficients $\tilde{E}_k(i)$ converge to
\begin{equation}
    \lim_{\gamma \to \infty} \tilde{E}_k(i) =
    \begin{cases}
        1, & \text{if } k = \displaystyle\arg\max_{j} E_j(i), \\
        0, & \text{otherwise}.
    \end{cases}
\end{equation}
Physically, the sharpened POVM converges pointwise to a projective measurement onto the dominant component.
\label{theorem:sharpness_theorem}
\end{theorem}

The sharpening transformation provides a continuous interpolation between unsharp and projective measurements. As established in Theorem~\ref{theorem:sharpness_theorem}, the parameter $\gamma$ controls the concentration of the POVM elements in the intensity basis. For $\gamma = 1$, the operators $\{E_k\}$ retain their Gaussian form, corresponding to an unsharp measurement where each intensity contributes probabilistically to multiple outcomes, resulting in smooth distributions and structure-preserving reconstructions. As $\gamma$ increases, the normalized elements
\begin{equation}
    E_k(i) \;\rightarrow\; \frac{E_k(i)^\gamma}{\sum_j E_j(i)^\gamma}
\end{equation}
become more concentrated around dominant components, reducing overlap and inducing localization. In the limit $\gamma \to \infty$, the measurement approaches a projective-valued measure (PVM), consistent with the trends observed in Fig.~\ref{fig:combined_analysis}. These results support Theorem~\ref{theorem:sharpness_theorem} and show that $\gamma$ provides a principled control over the trade-off between smoothing and localization. In cases where multiple indices attain the maximum value, the limiting distribution is supported on the set of maximizing indices. Unlike threshold-based methods, $\gamma$ provides a continuous interpolation between probabilistic smoothing and hard localization, avoiding abrupt intensity partitioning.

\subsection{Physical Realization on Quantum Hardware}

\begin{figure}
    \centering
    \begin{quantikz}
        \lstick{$\rho_x$} & \qw & \gate[wires=2]{U} & \qw & \qw \\
        \lstick{$|0\rangle$} & \qw &  & \meter{} & \rstick{$p_k(x)$} \cw \\
        & \cw
        &\rstick{$p_1(x)$} \\
        & \cw
        & \rstick{$p_2(x)$} \\
        & \cw
        & \vdots \\
        & \cw
        & \rstick{$p_K(x)$}
    \end{quantikz}
    \caption{Naimark dilation of the adaptive Gaussian POVM. A joint unitary $U$ acts on the system state $|I(x)\rangle$ and an ancilla initialized in $|0\rangle$, producing the state $\sum_k (M_k |I(x)\rangle)\otimes |k\rangle$. Measurement of the ancilla yields outcome $k$ with probability $p_k(x)=\mathrm{Tr}(E_k \rho_x)$.}
    \label{fig:qc}
\end{figure}
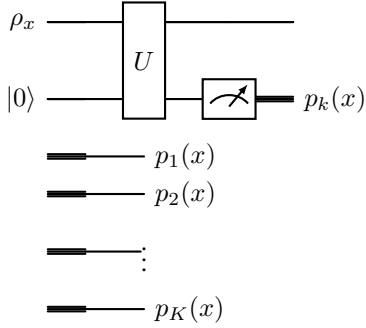

Although POVMs are not unitary operations, they can be realized through unitary evolution on an extended Hilbert space followed by projective measurement, as guaranteed by Naimark's dilation theorem \cite{peres1990neumark, holevo2011probabilistic,beneduci2017joint, pellonpaa2023naimark}. Given the POVM elements $\{E_k\}_{k=1}^K$, we define the corresponding Kraus operators
\begin{equation}
M_k = \sqrt{E_k} = \sum_{i=0}^{255} \sqrt{E_k(i)} \, |i\rangle \langle i|,
\end{equation}
where $E_k(i)$ are the diagonal elements of the POVM in the intensity basis. These operators satisfy the completeness relation
\begin{equation}
\sum_{k=1}^K M_k^\dagger M_k = \sum_{k=1}^K E_k = I.
\end{equation}

We define an isometry $V : \mathcal{H} \rightarrow \mathcal{H} \otimes \mathcal{H}_a$ by
\begin{equation}
V |\psi\rangle = \sum_{k=1}^K M_k |\psi\rangle \otimes |k\rangle,
\end{equation}
where $\{|k\rangle\}$ forms an orthonormal basis of the ancilla space $\mathcal{H}_a$. The completeness condition ensures that $V^\dagger V = I$, and hence $V$ can be extended to a unitary operator $U$ acting on the joint system:
\begin{equation}
U \left(|\psi\rangle \otimes |0\rangle\right) = V |\psi\rangle.
\end{equation}

For the computational basis states $|i\rangle$, this unitary implements the mapping
\begin{equation}
|i\rangle \otimes |0\rangle \;\longmapsto\; \sum_{k=1}^K \sqrt{E_k(i)} \, |i\rangle \otimes |k\rangle,
\end{equation}
which correlates the system with measurement outcomes stored in the ancilla. A projective measurement on the ancilla then yields outcome $k$ with probability
\begin{equation}
p_k = \mathrm{Tr}(E_k \rho),
\end{equation}
thus reproducing the desired POVM statistics. The circuit shown in Fig.~\ref{fig:qc} provides a schematic representation of this dilation. Since the operators $M_k$ are diagonal, the unitary $U$ reduces to controlled amplitude redistribution conditioned on the computational basis states. The implementation is compatible with near-term quantum devices since this structure enables effective deconstruction into elementary gates. It is important to emphasize that the measurement statistics, not the unitary evolution itself, are the source of the transformation in the current framework. The POVM is only embedded into a physically realizable process by the unitary $U$.

In the present framework, the POVM elements are diagonal in the computational (intensity) basis, enabling a structured realization. The corresponding dilation can be implemented via controlled operations conditioned on basis states, avoiding synthesis of a dense unitary and reducing circuit complexity with favorable scaling in $K$ and $\log d$. However, full-scale implementation for high-resolution images remains resource-intensive. Therefore, the framework is best interpreted as quantum-inspired for current hardware, or realized in hybrid schemes where probability evaluation is performed classically. Developing efficient implementations of such structured POVMs on near-term devices remains an open direction.

\section{Conclusion}\label{conclusion}

In this work, we proposed a probabilistic intensity remapping framework for grayscale images based on adaptive positive operator-valued measures (POVMs). By constructing data-driven measurement operators from Gaussian models of the image histogram, the method defines a continuous, structure-preserving mapping that replaces the rigid, piecewise-constant behavior of conventional thresholding approaches. The remapped intensities are obtained as expectation values, providing a principled probabilistic formulation of intensity transformation. A nonlinear sharpening parameter enables a continuous transition from soft probabilistic remapping to hard assignment, allowing explicit control over the trade-off between smoothing and structural localization. When compared to traditional thresholding-based techniques, experimental results show steady gains in PSNR, SSIM, and entropy, demonstrating the efficacy of the suggested strategy. All things considered, the framework admits a natural interpretation within quantum measurement theory and establishes probabilistic intensity remapping as a versatile and data-adaptive substitute for discrete histogram partitioning.

The suggested framework makes a number of recommendations for future quantum information processing research. The expectation-value formulation is in line with quantum machine learning paradigms based on observables and measurement outcomes, and the data-adaptive construction of POVMs offers a natural connection between measurement design and learning-based techniques. A technique to interpolate between probabilistic and deterministic behavior in tasks like classification and clustering is provided by the sharpening parameter $\gamma$, which permits a controlled shift between unsharp and projective regimes. Correlations and more expressive representations beyond intensity could be included in extensions to non-diagonal operators and higher-dimensional embeddings. These directions suggest that quantum measurement theory and statistical learning can be unified through measurement-induced transformations.

\ifCLASSOPTIONcaptionsoff
\newpage
\fi

\bibliographystyle{IEEEtran}

\bibliography{IEEE}

\section{Supplementary Information}
\begin{figure*}
    \centering
    \begin{subfigure}[t]{0.21\linewidth}
        \centering
        \includegraphics[width=\linewidth]{Images/100/100_km_rec_img_k4.png}
        \caption{k = 2}
    \end{subfigure}\hfill
    \begin{subfigure}[t]{0.21\linewidth}
        \centering
        \includegraphics[width=\linewidth]{Images/100/100_km_rec_img_k4.png}
        \caption{k = 4}
    \end{subfigure}\hfill
    \begin{subfigure}[t]{0.21\linewidth}
        \centering
        \includegraphics[width=\linewidth]{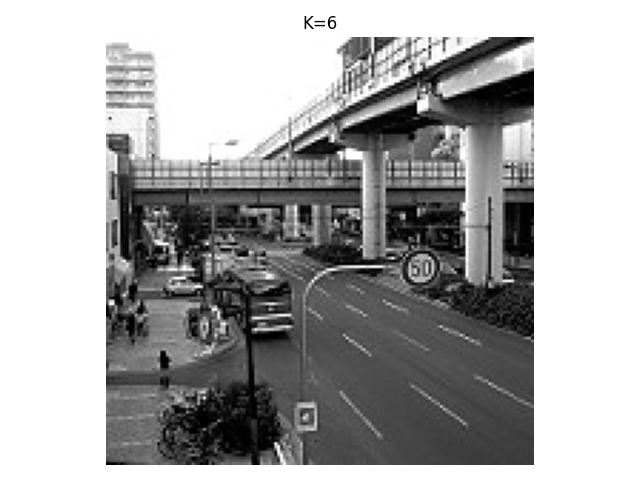}
        \caption{k = 6}
    \end{subfigure}\hfill
    \vspace{0.3em}
    \begin{subfigure}[t]{0.21\linewidth}
        \centering
        \includegraphics[width=\linewidth]{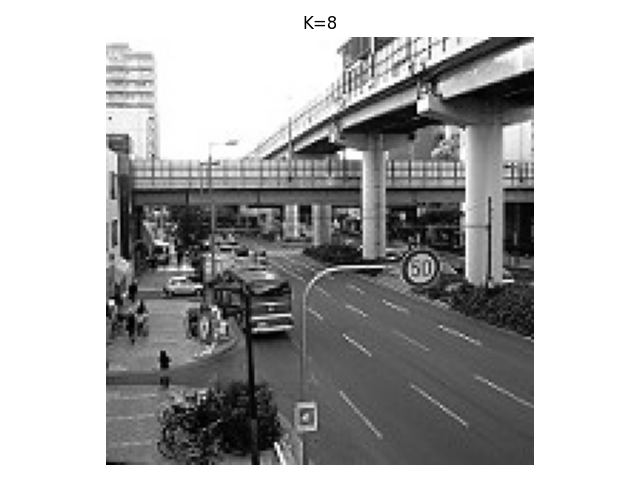}
        \caption{k= 8}
    \end{subfigure}
    \begin{subfigure}[t]{0.14\linewidth}
        \centering
        \includegraphics[width=\linewidth]{Images/100/100.jpg}
        \caption{Original}
    \end{subfigure}\hfill
    \caption{Robustness of the proposed model with K-means clustering. The first four images show reconstructed outputs for varying $k$, and the last image is the original.}
    \label{fig:km_k_2_8_100}
\end{figure*}

\begin{figure*}
    \centering
    \begin{subfigure}[t]{0.21\linewidth}
        \centering
        \includegraphics[width=\linewidth]{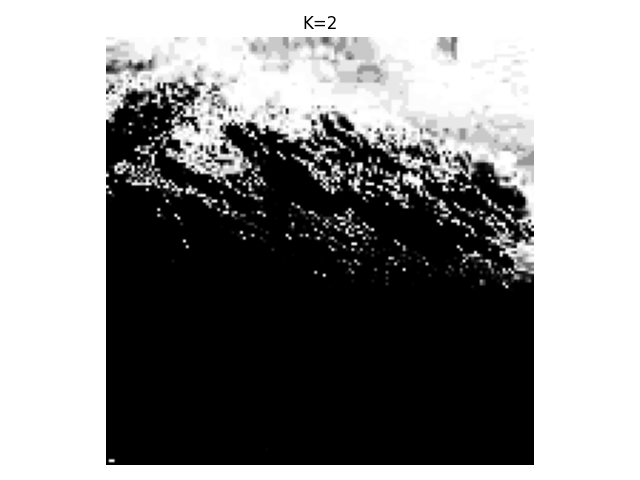}
        \caption{k = 2}
    \end{subfigure}\hfill
    \begin{subfigure}[t]{0.21\linewidth}
        \centering
        \includegraphics[width=\linewidth]{Images/1001/1001_gmm_rec_img_k4.png}
        \caption{k = 4}
    \end{subfigure}\hfill
    \begin{subfigure}[t]{0.21\linewidth}
        \centering
        \includegraphics[width=\linewidth]{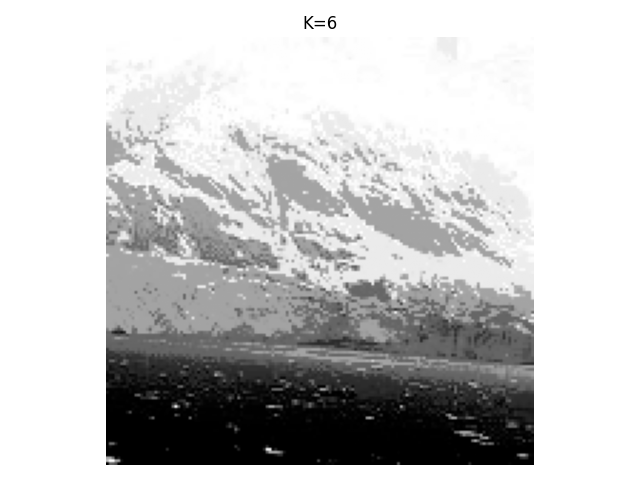}
        \caption{k = 6}
    \end{subfigure}\hfill
    \begin{subfigure}[t]{0.21\linewidth}
        \centering
        \includegraphics[width=\linewidth]{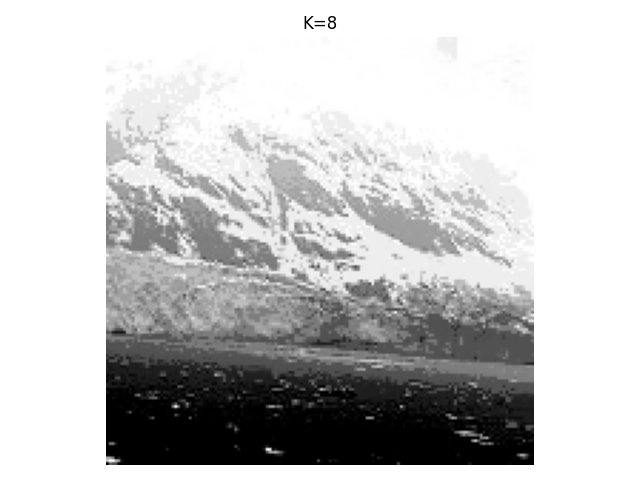}
        \caption{k= 8}
    \end{subfigure}
    \begin{subfigure}[t]{0.14\linewidth}
        \centering
        \includegraphics[width=\linewidth]{Images/1001/1001.jpg}
        \caption{Original}
    \end{subfigure}\hfill
    \caption{Robustness of the proposed model with Gaussian Mixture model clustering. The first four images show reconstructed outputs for varying $k$, and the last image is the original.}
    \label{fig:gmm_k_2_8_1001}
\end{figure*}

\subsection{Proof and Additional Discussion of Theorem 1}
\begin{proof}
For a pixel $(x,y)$ with intensity $i = I(x,y)$, we have
\begin{equation}
\hat{I}(x,y) = \sum_{k=1}^K \mu_k E_k(i).
\end{equation}

Since $\sum_k E_k(i) = 1$ and $E_k(i) \ge 0$, the reconstruction is a convex combination of the values $\{\mu_k\}$. Therefore,
\begin{equation}
    |\hat{I}(x,y) - i| = \left| \sum_{k=1}^K (\mu_k - i)E_k(i)\right|
    \leq \sum_{k=1}^K |\mu_k - i| E_k(i).
\end{equation}
By assumption, for each $i$ there exists $\mu_{k^*}$ such that $|\mu_{k^*} - i| \le \epsilon_K$, and as $K$ increases, the Gaussian construction ensures that the weights $E_k(i)$ concentrate around such indices. Hence,
\begin{equation}
|\hat{I}(x,y) - i| \leq \epsilon_K,
\end{equation}
with $\epsilon_K \to 0$ as $K \to \infty$. Therefore,
\begin{equation}
    \hat{I}(x,y) \to I(x,y),
\end{equation}
which completes the proof.
\end{proof}

By varying the number of Gaussian components $K$, the framework's adaptive behavior can be better understood. Figs.~\ref{fig:gmm_k_2_8_1001} and \ref{fig:km_k_2_8_100} demonstrate how raising $K$ results in increasingly better reconstructions. The POVM creates a coarse partition of the intensity space for small $K$, which limits the dynamic range and produces smoother outputs because of stronger averaging. Larger $K$ results in a finer coverage from the Gaussian components, which enhances contrast and structural detail by enabling the measurement probabilities to more precisely capture local intensity variations. Quantitative trends in PSNR, SSIM, and entropy support this observation. Both GMM- and KMeans-based constructions exhibit this improvement, although GMM generally yields smoother and more statistically consistent results due to its probabilistic modeling, whereas KMeans relies on hard clustering. These observations align with the convex combination structure of the reconstruction, where increasing $K$ enhances representational capacity.

\subsection{Proof and Additional Discussion of Theorem 2}
\begin{proof}
Fix $i \in \{0,1,\dots,255\}$ and define
\begin{equation}
M(i) = \max_{j} E_j(i).
\end{equation}
Let $k^*$ be an index such that $E_{k^*}(i) = M(i)$. Then, for any $k \neq k^*$, we have
\begin{equation}
    0 \leq \frac{E_k(i)}{E_{k^*}(i)} < 1.
\end{equation}

Now consider the sharpened coefficients:
\begin{equation}
\tilde{E}_k(i) = \frac{E_k(i)^\gamma}{\sum_{j=1}^K E_j(i)^\gamma}.
\end{equation}

For $k = k^*$,
\begin{align}
\tilde{E}_{k^*}(i) 
&= \frac{E_{k^*}(i)^\gamma}{E_{k^*}(i)^\gamma + \sum_{j \neq k^*} E_j(i)^\gamma} \\
&= \frac{1}{1 + \sum_{j \neq k^*} \left(\frac{E_j(i)}{E_{k^*}(i)}\right)^\gamma}.
\end{align}

Since $\frac{E_j(i)}{E_{k^*}(i)} < 1$ for all $j \neq k^*$, it follows that
\begin{equation}
\left(\frac{E_j(i)}{E_{k^*}(i)}\right)^\gamma \to 0 \quad \text{as } \gamma \to \infty.
\end{equation}
Hence,
\begin{equation}
    \lim_{\gamma \to \infty} \tilde{E}_{k^*}(i) = 1.
\end{equation}

For $k \neq k^*$,
\begin{align}
\tilde{E}_k(i) 
&= \frac{E_k(i)^\gamma}{E_{k^*}(i)^\gamma + \sum_{j \neq k^*} E_j(i)^\gamma} \\
&= \frac{\left(\frac{E_k(i)}{E_{k^*}(i)}\right)^\gamma}{1 + \sum_{j \neq k^*} \left(\frac{E_j(i)}{E_{k^*}(i)}\right)^\gamma}.
\end{align}
Again, since $\frac{E_k(i)}{E_{k^*}(i)} < 1$, we obtain
\begin{equation}
    \lim_{\gamma \to \infty} \tilde{E}_k(i) = 0.
\end{equation}

Therefore, the sharpened coefficients converge pointwise to a one-hot distribution concentrated on the maximizing index $k^*$, completing the proof.
\end{proof}

The effect of the sharpening parameter $\gamma$ can be understood both visually and quantitatively. Larger $\gamma$ causes the measurement to become essentially projective, which reduces intensity variability and gives rise to piecewise-constant or edge-dominated representations. Increasing sharpness suppresses fine-grained structure by reducing probability overlap, which is consistent with quantum measurement theory. The reconstruction is equivalent to the expectation value of an observable evaluated on a perturbed state when noise is present because changes in input intensities spread across the measurement probabilities. Even though the framework isn't specifically made for denoising, $\gamma$ offers a logical way to manage the trade-off between localization and smoothing in noisy environments.

\end{document}